\documentclass[11pt, a4paper]{amsart}
\usepackage[dvips]{epsfig}
\usepackage{amsmath}
\usepackage{amssymb}
\usepackage{amsfonts}
\usepackage{amsthm}
\usepackage{amsbsy}
\usepackage{amsgen}
\usepackage{amscd}
\usepackage{amsopn}
\usepackage{amstext}
\usepackage{amsxtra}
\usepackage{enumerate}
\usepackage{dsfont}
\usepackage{hyperref}

\newtheorem{theorem}{Theorem}[section]
\newtheorem{proposition}[theorem]{Proposition}
\newtheorem{lemma}[theorem]{Lemma}
\newtheorem{corollary}[theorem]{Corollary}
\theoremstyle{definition}

\newtheorem{remark}[theorem]{Remark}

\begin{document}
\title{On Nonlinear Functionals of Random Spherical Eigenfunctions}
\author{Domenico Marinucci}
\thanks{Research of D.M. is supported by the ERC Grant 277742 \emph{Pascal}.
}
\address{University of Rome Tor Vergata, Rome, Italy}
\author{Igor Wigman}
\thanks{Research of I.W. is supported by an EPSRC grant under the First
Grant Scheme.}
\address{King's College London, UK}
\date{\today }
\maketitle

\begin{abstract}
We prove Central Limit Theorems and Stein-like bounds for the asymptotic
behaviour of nonlinear functionals of spherical Gaussian eigenfunctions. Our
investigation combine asymptotic analysis of higher order moments for
Legendre polynomials and, in addition, recent results on Malliavin calculus
and Total Variation bounds for Gaussian subordinated fields. We discuss
application to geometric functionals like the Defect and invariant
statistics, e.g. polyspectra of isotropic spherical random fields. Both of
these have relevance for applications, especially in an astrophysical
environment.

\begin{itemize}
\item \textbf{Keywords and Phrases: }Gaussian Eigenfunctions,High Energy
Asymptotics, Central Limit Theorem, Total Variation Bounds

\item \textbf{AMS Classification: }60G60; 42C10, 60D05, 60B10

\item \textbf{PACS: }02.50Ey, 02.30Nw, 02.30Px
\end{itemize}
\end{abstract}

\section{Introduction}

\subsection{Motivation}

Much effort has been recently devoted to the analysis of spherical Gaussian
eigenfunctions (to be defined below). These random fields are the
Fourier components in spectral representation expansions for general
spherical Gaussian fields, see for instance \cite{adlertaylor}, \cite{leosa}%
, \cite{MaPeCUP}, and \cite{leonenko2}, \cite{mal} for extensions; in view
of this, their study is of obvious relevance in connection with the
statistical analysis of spherical data. Namely, the analysis of these
components and their polynomial transforms (the so-called \emph{polyspectra}%
) is now one of the leading themes in modern Cosmology, in particular in the
growing area of Cosmic Microwave Background (CMB) data analysis: we refer
for instance to \cite{bartolo}, \cite{durrer}, \cite{lewis}, \cite{MaPeCUP}
and the references therein.

In short and somewhat vague terms, CMB represents a relic electromagnetic
radiation from the so-called \emph{age of recombination}, e.g. the
cosmological epoch when free electrons and protons combined to form stable
hydrogen atoms; this is now reckoned to have occurred some $3.7\times 10^{5}$
years after the Big Bang, i.e., some $1.3\times 10^{10}$ years ago. As such,
CMB\ radiation is universally recognized as a goldmine of information on
primordial epochs, and its analysis has drawn enormous theoretical and
experimental efforts: we refer for instance to http://map.gsfc.nasa.gov/ and
http://www.rssd.esa.int/Planck for the most relevant ongoing experimental
activity.

Among many, one of the leading current themes is the analysis of
non-Gaussianity of CMB data; for this purpose much effort has been dedicated
to the investigation of the properties of the polyspectra mentioned above,
e.g. polynomial transforms of sample spherical Fourier components. In this
paper we provide a full characterization for the asymptotic behaviour of
these transforms, proving Central Limit Theorem results under rather broad
assumptions.

\vspace{2mm} From a different perspective, the analysis of high frequency
(or high energy) Laplace eigenfunctions is a classical subject in Analysis
and Mathematical Physics with much progress in understanding, both
rigorously and conjecturally, in the recent years. According to Berry's
universality conjecture ~\cite{Berry 1977}, the \emph{deterministic} high
energy eigenfunctions on "generic" manifolds are represented by \emph{random}
monochromatic waves, an equivalent planar model to our spherical Gaussian
eigenfunctions.

One in particular interesting aspect of those is their \emph{nodal structure}
~\cite{BGS} (e.g. the zero set and the number of its connected components or
ovals), especially in light of their conjectural relations to percolation
and SLE ~\cite{BS}. Some results on the geometry of the nodal structure
include studying the number of the ovals (equivalently the number of \emph{%
nodal components} i.e. connected components of the zero set complement) ~%
\cite{NS}, and the total length of the nodal line ~\cite{Wig}. The so-called
Defect (i.e. the difference between ``hot" and ``cold" regions for the
eigenfunctions, see below or Section \ref{sec:defect}) of spherical Gaussian
eigenfunctions was addressed in our earlier work \cite{MaWi2}.

\subsection{Statement of the main results}

Let $\Delta _{\mathbb{S}^{2}}$ denote the usual Laplacian on the $2$%
-dimensional unit sphere $\mathcal{S}^{2}$, and consider the Gaussian random
spherical eigenfunctions $f_{l}$, $l\in\mathbb{Z}_{>0}$, satisfying
\begin{equation*}
\Delta _{\mathbb{S}^{2}}f_{l}=-l(l+1)f_{l}\text{ , }\mathbb{E}\left[ f_{l}(x)
\right] =0\text{ , }\mathbb{E}\left[ f_{l}(x)^{2}\right] =1\text{.}
\end{equation*}%
It is well known that $f_{l}$ can be expanded into the orthonormal basis of
spherical harmonics $\left\{ Y_{lm}\right\} $%
\begin{equation*}
f_{l}=\sum_{m=-l}^{l}a_{lm}Y_{lm}\text{ , }\mathbb{E}\left[ a_{lm}\overline{a%
}_{lm^{\prime }}\right] =\frac{4\pi }{2l+1}\delta _{m}^{m^{\prime }},
\end{equation*}
(where the coefficients $a_{lm}$ Gaussian i.i.d.). The random field $f_{l}$
is isotropic or rotation invariant, meaning that for any rotation $g\in
SO(3) $ the distribution of $f_{l}(\cdot)$ is equal to the distribution of $%
f_{l}(g\,\cdot)$ (equivalently, for every $k\in \mathbb{N}$ and $%
x_{1},\ldots ,x_{k}\in \mathbb{S}^{2}$, the distribution of the random
vector $\left( f_{l}(x_{1}),\ldots ,f_{l}(x_{k})\right) $ equals to the
distribution of $\left( f_{l}(g\cdot x_{1}),\ldots ,f_{l}(g\cdot
x_{k})\right) $). Equivalently, $f_{l}$ is centred Gaussian, with the
covariance function
\begin{equation*}
r_{l}(x,y):=\mathbb{E}[f_{l}(x)\cdot f_{l}(y)]=P_{l}(\left\langle
x,y\right\rangle )\text{ , }x,y\in S^{2},
\end{equation*}%
where $\left\langle \cdot,\cdot\right\rangle $ is the Euclidean inner product, so
that $\vartheta (x,y):=\arccos (\left\langle x,y\right\rangle )$ is the
angular distance; here, $P_{l}:[-1,1]\rightarrow \mathbb{R}$ are the usual
Legendre polynomials, see the Appendix for more details.

The purpose of this paper is to study the asymptotic behaviour of random
quantities such as
\begin{equation*}
h_{l;q}=\int_{S^{2}}H_{q}(f_{l}(x))dx\text{ ,}
\end{equation*}%
where $H_{q}(.)$ are the Hermite polynomials satisfying the differential
equations%
\begin{equation*}
H_{m}^{\prime }(x)=mH_{m-1}(x)\text{ , }\mathbb{E}\left[ H_{m}(Z)\right] =0%
\text{ , }Z\sim \mathcal{N}(0,1)\text{ .}
\end{equation*}%
As mentioned above, these nonlinear transforms are of interest by themselves
as statistical functionals, in connections to the analysis of angular
polyspectra of spherical random fields (more details to be provided below).
Importantly, these statistics are the basic building blocks for the analysis
of general nonlinear functionals of Gaussian eigenfunctions. More precisely,
it is a well-known general fact (see for instance \cite{peccatitaqqu}) that
the $L^{2}$ space of square-integrable nonlinear transforms of Gaussian
eigenfunctions can be expanded (in the $L^{2}$ sense) as%
\begin{equation}
G(f_{l})=\sum_{q=0}^{\infty }\frac{J_{q}(G)}{q!}H_{q}(f_{l})\text{ , }%
\mathbb{E}\left[ G^{2}(f_{l})\right] <\infty \text{ , }J_{q}(G):=\mathbb{E}%
\left[ G(f_{l})H_{q}(f_{l})\right] \text{ .}  \label{L2exp}
\end{equation}%
As a consequence, the analysis of averaged statistics of the form
\begin{equation*}
\mathcal{G }(f_{l})=\int_{S^{2}}G(f_{l})dx
\end{equation*}
for ``generic" $G$ will directly follow from Central Limit Theorem results on $\left\{
h_{l;q}\right\} .$ To establish the latter, we first need to investigate the
asymptotic behaviour, as $l\rightarrow \infty ,$ for the variances $
Var(h_{l;q}).$ Note that when both $q$ and $l$ are odd, the statistics
$\left\{ h_{l;q}\right\} $ are identically zero for the symmetry properties
of Legendre polynomials, e.g.%
\begin{equation*}
P_{l}(x)=(-1)^{l}P_{l}(-x),
\end{equation*}%
whence integrals of odd polynomials over the sphere are identically zero. To
simplify the discussion, throughout the sequel we shall consider all limits
only for even multipoles $l.$ Under this condition, the asymptotic behaviour
of these variances was investigated in \cite{MaWi2}, where it was shown that,
for $q=3$ and $q\geq 5$
\begin{eqnarray*}
Var(h_{q;l}) &=&(4\pi )^{2}q!\int_{0}^{\pi /2}P_{l}^{q}(\cos \vartheta )\sin
\vartheta d\vartheta \sim (4\pi )^{2}q!\frac{c_{q}}{l^{2}}\text{ ,} \\
c_{q} &=&\int_{0}^{\infty }\psi J_{0}(\psi )^{q}d\psi \geq 0\text{ .}
\end{eqnarray*}%
For $q=2,4,$ the order of magnitude of the corresponding variances is larger
(see Lemma \ref{lem:varhl4=4thmoment=logl/l^2}):
\begin{equation*}
Var(h_{q;l})\approx \left\{
\begin{array}{c}
\frac{1}{l}\text{ , for }q=2 \\
\frac{\log l}{l^{2}}\text{ , for }q=4%
\end{array}%
\right. \text{ .}
\end{equation*}%
The constants $c_{q}$ are immediately seen to be strictly positive for all
even values of $q.$ For odd values, we conjecture this to be always the
case; a formal proof is left for future research. Therefore, the statement
of our main result will entail this condition explicitly:

\begin{theorem}
\label{thm:hlq asymp gauss} For all $q$ such that $c_{q}>0,$ we have
\begin{equation*}
\frac{h_{2l;q}}{\sqrt{Var(h_{2l;q})}}\rightarrow _{d}\mathcal{N}(0,1)\text{, as }l\rightarrow \infty \text{ .}
\end{equation*}%
\bigskip
\end{theorem}

Theorem \ref{thm:hlq asymp gauss} is a building block for a more general
claim: under minimal regularity conditions (i.e. the existence of at least
one nonzero coefficient $J_{q}(G)$ corresponding to a nonzero term $c_{q}>0$%
) we shall have a CLT for square integrable nonlinear functionals of
spherical Gaussian eigenfunctions, i.e. for any $\mathcal{G}(f_{l})=\int_{%
\mathbb{S}^{2}}G(f_{l})dx$ such that $\mathbb{E}\mathcal{G}%
^{2}(f_{l})=:\sigma _{G}^{2}(l)<\infty $,
\begin{equation*}
\frac{\mathcal{G}(f_{2l})-\mathbb{E}\left[ \mathcal{G}(f_{2l})\right] }{%
\sigma _{G}(2l)}\rightarrow _{d}\mathcal{N}(0,1)\text{, as }l\rightarrow
\infty .
\end{equation*}

In fact we will prove the following stronger result: for $q\neq 4$,
\begin{equation}
d_{TV}\left( \frac{h_{2l;q}}{\sqrt{Var(h_{2l;q})}},\mathcal{N}(0,1)\right)
=O\left( l^{-\delta _{q}}\right)  \label{bound1}
\end{equation}%
for some $\delta _{q}>0$, where $d_{TV}(.,.)$ denotes the usual total
variation distance of random variables
\begin{equation*}
d_{TV}(X,Y)=\sup_{A\in \mathcal{B}(\mathbb{R})}\left\vert \Pr (X\in A)-\Pr
(Y\in A)\right\vert ,
\end{equation*}%
with the Borel $\sigma $-field $\mathcal{B}(\mathbb{R})$. For $q=4$ the rate
of convergence \eqref{bound1} is of slower \emph{logarithmic} order.

\subsection{On the proofs of the main results and outline of the paper}

The ideas behind our main argument can be summarized as follows. Because $%
f_{l}(\cdot)$ is a Gaussian field, for any fixed $x\in \mathbb{S}^{2}$,
$H_{q}(f_{l}(x))$
belongs to the so-called $q$-th order Wiener chaos generated by the
Gaussian measure governing $f_{l}(\cdot)$ (see \cite{peccatitaqqu}), and so does
any linear transform, including $h_{l;q}.$ As a consequence of the
Nourdin-Peccati Theorem for Stein-Malliavin approximations of Gaussian
subordinated random variables (see for instance \cite{noupebook}, Theorem
5.26), the following bound holds for each even $l,$ $q\geq 2:$%
\begin{equation}
d_{TV}\left( \frac{h_{l;q}}{\sqrt{Var(h_{l;q})}},\mathcal{N}(0,1)\right)
\leq 2\sqrt{\frac{q-1}{3q}\left( \frac{cum_{4}\left( h_{l;q}\right) }{%
Var^{2}(h_{l;q})}\right) }\text{ ;}  \label{noupebou2}
\end{equation}
here, $cum_{4}\left( Y\right) $ is the $4$th order cumulant of $Y,$ see \cite%
{noupebook}, \cite{peccatitaqqu} for more discussion on these points.

The latter bound shows that if we prove that
\begin{equation}  \label{eq:cum4=o(var^2)}
cum_{4}\left( h_{l;q}\right) = o_{l\rightarrow\infty} \left(
Var^{2}(h_{l;q})\right),
\end{equation}
the Central Limit Theorem for $h_{l;q}$ (where $q$ is fixed and $%
l\rightarrow\infty$) will follow. The bound \eqref{eq:cum4=o(var^2)} for $%
q\ge 5$ is proved along Section \ref{q>=5}. Here we first express the $4$th
order cumulant as an integral over $(\mathcal{S}^{2})^{4}$, using the
well-known \emph{Diagram Formula} (see Section \ref{sec:Diagram Formula}).
This will allow us to obtain the desired bound \eqref{eq:cum4=o(var^2)} via
a tricky multiple application of the Cauchy-Schwartz inequality (Proposition %
\ref{prop:cums are small}, whose proof in Section \ref{sec:cum small proof}
takes on most of Section \ref{q>=5}); to this end we will divide the domain
of integration into the ``local" and ``global" ones (for the definitions of
various ranges see Section \ref{sec:Loc Glob}). A more detailed explanation
of the proof of Proposition \ref{prop:cums are small} may be found in
Section \ref{sec:on proof of cum bnd}.

For $q=3,4$ proving the bound \eqref{eq:cum4=o(var^2)} will require special
arguments, presented in Section \ref{q=3,4}. While the case $q=3$ was
already covered earlier in \cite{M2008}, for $q=4$ the asymptotic analysis
requires the evaluation of so-called Gaunt integrals, connecting moments of
Legendre polynomials to Wigner's and Clebsch-Gordan coefficients (see \cite%
{MaPeCUP}, \cite{VMK}).

Various applications of Theorem \ref{thm:hlq asymp gauss} for the Defect
statistics and general polyspectra are discussed in Section \ref%
{applications}. The basic idea is to use the Central Limit Theorem for $%
h_{l;q}$ to establish asymptotic Gaussianity of \emph{finite-order} generic
polynomial sequences, and then exploit the expansion \eqref{L2exp} of
arbitrary functionals (e.g. the Defect) to prove that the distribution of
any such functional can be asymptotically approximated by means of a
finite-order expansion.

\subsubsection{Some conventions}

In this manuscript, given any two positive sequences $a_{n},b_{n}$ we shall
write $a_{n}\approx b_{n}$ if there exist two positive constants $%
c_{1},c_{2} $ such that $c_{1}a_{n}\leq b_{n}\leq c_{2}a_{n},$ for all $%
n=1,2,...,$ and $a_{n}\sim b_{n}$ if $\lim_{n\rightarrow \infty
}a_{n}/b_{n}=1.$ Also, we shall write $a_{n}\ll b_{n}$ or $a_{n}=O(b_{n})$
when the sequence $a_{n}/b_{n}$ is asymptotically bounded; we write $\mu
(dx) $ for the usual Lebesgue measure on the unit sphere, so that $\int_{%
\mathbb{S}^{2}}\mu (dx)=4\pi$.

\subsection{Acknowledgements}

The authors would like to thank Giovanni Peccati, Ze\'{e}v Rudnick and
Mikhail Sodin for many stimulating and fruitful discussions about the
subjects raised in the present manuscript. A substantial part of this
research was done during the second author's visit to University of Rome
\textquotedblleft Tor Vergata", and he would like to acknowledge the
exemplary hospitality of the institution, and the generous financial support.

\section{The central limit theorem for $h_{l;q}$, $q\geq 5$ \label{q>=5}}

\subsection{Some preliminaries}

\label{sec:Diagram Formula}

We shall now focus on fourth order cumulants for Hermite transforms for
arbitrary $q\geq 5.$ First we need to recall some well-known background
material on the so-called Diagram Formula (see for instance \cite%
{Sodin&Tsirelson}, \cite{peccatitaqqu}, \cite{MaPeCUP} or \cite%
{nourdinpeccati} for recent textbook references).

Fix a set of integers $\alpha _{1},...,\alpha _{p}.$ A diagram is a graph
with $(\alpha _{1}+...+\alpha _{p})$ vertexes labelled by $1,...,p,$ ($%
\alpha _{1}$ vertexes are labelled by $1,$ $\alpha _{2}$ vertexes are
labelled by 2...) such that each vertex has degree 1, i.e. the edges have no
common endpoints. We can view the vertexes as belonging to $p$ different
rows and the edges may connect only vertexes with different labels, i.e.
there are no flat edges on the same row. The set of all such graphs $\gamma $
is denoted by $\Gamma (\alpha _{1},...,\alpha _{p});$ we write $\Gamma
_{C}(\alpha _{1},...,\alpha _{p})$ for graphs that are connected, i.e. it is
not possible to partition the vertexes into two subsets $A$ and $B$ such
that no edge connects a vertex in $A$ with one in $B.$

Given a diagram $\gamma$, let
\begin{equation*}
\underline{\eta}(\gamma)=(\eta _{ij}(\gamma ))\in \mathbb{Z}^{\binom{p}{2}}
\end{equation*}
be the vector whose $\binom{p}{2}$ elements $\eta_{ij}(\gamma)$ ($i < j$)
are the number of edges between $i$ and $j$ in the graph $\gamma$. The
vector $\underline{\eta}$ satisfies
\begin{equation}  \label{eq:sum etaij=2q}
\sum_{i,j}\eta _{ij}=2q,
\end{equation}
and, moreover, the multiplicities of opposite edges equal, as the following
lemma states.

\begin{lemma}
\label{lem:etaij=etai'j'} Let $\gamma \in \Gamma _{C}(q,q,q,q)$ with
arbitrary $q\geq 1$, and $\underline{\eta }=\underline{\eta }(\gamma )$. Let
$e=(i,j)$ any edge in $\gamma $ and $e^{\prime
}=(i^{\prime },j^{\prime })$ the unique edge with vertexes disjoint with
$e$, so that $\{i,j,i^{\prime },j^{\prime }\}=\{1,2,3,4\}$. Then
$\eta _{e}=\eta _{e^{\prime }}$.
\end{lemma}

\begin{proof}
The statement of the lemma follows immediately from the fact that, by
reordering the vertexes corresponding to the same label (i.e. in the same
row) if necessary, we may assume that all the edges are between vertexes in
the same column.
\end{proof}

\subsection{Cumulants of $h_{l;q}$}

With notation as above, the well-known Diagram Formula (see e.g.~\cite%
{peccatitaqqu}) provides the following neat expression for the cumulants of
our statistics $\left\{ h_{l;q}\right\} ,$ namely
\begin{equation*}
cum\left\{ h_{l;q_{1}},...,h_{l;q_{p}}\right\} =cum\left\{
\int_{S^{2}}H_{q_{1}}(f_{l}(x_{1}))dx_{1},...,%
\int_{S^{2}}H_{q_{p}}(f_{l}(x_{p}))dx_{p}\right\}
\end{equation*}%
\begin{eqnarray*}
&=&\int\limits_{(S^{2})^{p}}cum\left\{
H_{q_{1}}(f_{l}(x_{1})),...,H_{q_{p}}(f_{l}(x_{p}))\right\} dx_{1}...dx_{p}
\\
&=&\int\limits_{(S^{2})^{p}}\sum\limits_{\gamma \in \Gamma
_{C}(q_{1},...,q_{p})}\prod\limits_{(i,j)\in \gamma} \mathbb{E}%
\left\{ f_{l}(x_{i})f_{l}(x_{j})\right\} dx_{1}...dx_{p} \\
&=&\sum\limits_{\gamma \in \Gamma
_{C}(q_{1},...,q_{p})}\int\limits_{(S^{2})^{p}}\prod\limits_{(i,j)}\left(
\mathbb{E}\left\{ f_{l}(x_{i})f_{l}(x_{j})\right\} \right) ^{\eta
_{ij}(\gamma )}dx_{1}...dx_{p}\text{ , }
\end{eqnarray*}%
(recall that $\underline{\eta }(\gamma )$ is the vector of multiplicities of
edges in $\gamma $ as above).

We constraint ourselves to the $4$th order cumulants of the form $%
cum(h_{l;q},h_{l;q},h_{l;q},h_{l;q})$; in this case the expression above
simplifies to
\begin{equation}  \label{eq:cum as sum mom}
cum(h_{l;q},h_{l;q},h_{l;q},h_{l;q}) = \sum_{\gamma \in \Gamma
_{C}(q,q,q,q)} M(\underline{\eta}(\gamma)),
\end{equation}
where for a vector $\eta\in \mathbb{Z}_{\ge 0}^{6}$ we set
\begin{equation}  \label{eq:C mom def}
M(\underline{\eta}) = \int\limits_{(S^{2})^{4}} \prod\limits_{i < j}
P_{l}(\langle x_{i},x_{j}\rangle)^{\eta _{ij}}dx,
\end{equation}
with the shortcut $dx=dx_{1}\cdot\ldots\cdot dx_{4}$. Since
\eqref{eq:cum as
sum mom} is a finite summation of the $C(\underline{\eta})$ with $\underline{%
\eta}$ corresponding to a connected diagram $\gamma\in \Gamma _{C}(q,q,q,q)$
it then remains to bound each of terms as above separately.

We know ~\cite{MaWi}, Lemma 5.2, that for $q\geq 5$,
\begin{equation*}
Var(h_{l;q})\sim \frac{c_{q}}{l^{2}}
\end{equation*}%
and aim at proving the bound
\begin{equation*}
cum(h_{l;q},h_{l;q},h_{l;q},h_{l;q})=o\left( \frac{1}{l^{4}}\right) ,
\end{equation*}%
(or stronger); this is sufficient for the central limit theorem by \ref%
{noupebou2}. The following proposition implies this estimate bearing in mind %
\eqref{eq:cum as sum mom}.

\begin{proposition}
\bigskip \label{prop:cums are small} For $q\geq 5$ and $\underline{\eta }=%
\underline{\eta }(\gamma )$ with $\gamma $ a connected diagram, one has
\begin{equation*}
|M(\underline{\eta })|=%
\begin{cases}
O\left( \frac{(\log {l})^{2}}{l^{4\frac{1}{5}}}\right) & q=5 \\
O\left( \frac{(\log {l})^{2}}{l^{4\frac{2}{7}}}\right) & q=6 \\
O\left( \frac{(\log {l})^{3/2}}{l^{4+\frac{q-6}{2q-3}}}\right) & q\geq 7%
\end{cases}%
.
\end{equation*}%
In particular, for every $l\geq 5$,
\begin{equation*}
|M(\underline{\eta })|=o\left( \frac{1}{l^{4}}\right) .
\end{equation*}
\end{proposition}

The methods of the present section also give a useful (i.e. smaller than the
square of the variance, sufficient for central limit theorem) upper bound
for $M(\underline{\eta})$ with $\underline{\eta}$ corresponding to $q=4$. It
is however weaker than the precise asymptotics of Lemma \ref{lem:cum bnd q=4}%
, which is the reason for a special treatment we gave to $q=4$.

\begin{remark}
In fact, it has been recently established by G.Peccati and coauthors that
only a subset of terms as above need to be bounded to establish the CLT -
those corresponding to circular diagrams (i.e. diagrams, all of whose rows
are linked with precisely two other rows, see e.g. \cite{peccatitaqqu},
Proposition 11.2); the latter were easier to bound by our earlier methods.
We will however not use this observation as our present methods can cope
with arbitrary terms, though resulting in slightly weaker upper bounds.
\end{remark}

\subsection{Proof of Proposition \protect\ref{prop:cums are small}}

\label{sec:cum small proof}

\subsubsection{On the proof}

\label{sec:on proof of cum bnd}

One observes that for a ``generic" point, each of the $6$ terms $%
P_{l}(\langle x_{i},x_{j}\rangle)$ in the integrand product in
\eqref{eq:C
mom def} is bounded by $\frac{1}{\sqrt{l}}$ as in
\eqref{eq:Legendre decay
sqrt}. Therefore, unless at least one of the $6$ angles involved is small,
the integrand is of order $O\left( \frac{1}{l^{q}} \right)$, and the total
contribution to the integral in \eqref{eq:C mom def} for $q\ge 5$ \emph{%
should} be of order smaller than $\frac{1}{l^{4}}$, sufficient for the CLT%
\footnotemark .

\footnotetext{%
This heuristics is not entirely correct, as we believe the correct order of
magnitude of the $4$th order cumulant to be proportional to $\frac{1}{l^{5}}$
rather than $\frac{1}{l^{q}}$. It means that the regime where at least one
angle is small does contribute to the integral.}

To quantify the latter statement, we introduce a small parameter $\varepsilon =
\varepsilon(l)$ and separate the domain of integration $(\mathcal{S}^{2})^{4}$
into the set $\mathcal{L}(\varepsilon)$ of points $x \in (\mathcal{S}^{2})^{4}$
all of whose angles $\vartheta (x_{i},x_{j})$ are greater than $\varepsilon$,
and its complement ($\vartheta(x_{i},x_{j})\le \varepsilon$ for least one pair
of indexes $(i,j)$), see Section \ref{sec:Loc Glob}. We call the
contribution of the latter subdomain ``local" and, analogously, the former's
contribution, ``global". The global and local contribution are bounded in
sections \ref{sec:glob contribution} and \ref{sec:loc contribution}, whence
it remains to choose the optimal parameter $\varepsilon$ the arises from the
tradeoff we get from those bounds (Section \ref{sec:concl proof}).

To bound both the global and local case one employs the following
observation. It is possible to decrease the number of different angles $%
\vartheta (x_{i},x_{j})$ involved in the integral \ref{sec:Loc Glob} by
applying the Cauchy-Schwartz inequality, multiple times if necessary. The
upshot is that in integrals like
\begin{equation*}
\int P_{l}(\cos\theta (\langle x_{1},x_{2}\rangle )^{s_{1}} P_{l}(\cos\theta
(\langle x_{3},x_{4}\rangle )^{s_{2}} dx
\end{equation*}
the variables split and we end up evaluating moments of individual Legendre
polynomials (that were readily evaluated, see \eqref{eq:varhl4=4thmoment=logl/l^2}
and \eqref{eq:int(P^2)<<1/l}). Therefore, when we apply the
Cauchy-Schwartz inequality to reduce the number of different angles, it is
beneficial to pair up angles corresponding to disjoint edges in the diagram.
For $q\ge 7$ these kind of observations, combined with uniform estimate %
\eqref{eq:Pl(theta)<<sqrt(lepsilon) glob} valid on $\mathcal{L}(\varepsilon)$ and the
small measure of its complement, are sufficient for our purposes. The cases $%
q=6,5$ are a bit more subtle as in these cases we will have to exploit the
special structure of vectors $\underline{\eta}$ corresponding to connected
diagrams.

\subsubsection{Global and local terms}

\label{sec:Loc Glob}

Recall that for a $\underline{\eta }\in \mathbb{Z}_{\geq 0}^{6}$ we defined $%
C(\underline{\eta })$ as in \eqref{eq:C mom def}. For a small parameter $%
\varepsilon =\varepsilon (l)>0$ we decompose the domain of integration in %
\eqref{eq:C mom def} as following:
\begin{equation*}
(\mathbb{S}^{2})^{4}=\mathcal{L}(\varepsilon )\cup \mathcal{L}(\varepsilon
)^{c},
\end{equation*}%
where
\begin{equation*}
\mathcal{L}(\varepsilon ):=\left\{ (x_{1},\ldots x_{4})\in (\mathbb{S}%
^{2})^{4}:\:\forall i<j,\,\vartheta (x_{i},x_{j})>\varepsilon \right\}
\end{equation*}%
and $\mathcal{L}(\varepsilon )^{c}$ is its complement. We then write
\begin{equation*}
M(\underline{\eta })=M_{glob}(\underline{\eta };\varepsilon )+M_{loc}(%
\underline{\eta };\varepsilon ),
\end{equation*}%
where
\begin{equation}
M_{glob}(\underline{\eta };\varepsilon )=\int\limits_{\mathcal{L}%
}\prod\limits_{i<j}P(\langle x_{i},x_{j}\rangle )^{\eta _{ij}}dx
\label{eq:M glob def}
\end{equation}%
is the \textquotedblleft global contribution" and
\begin{equation*}
M_{loc}(\underline{\eta };\varepsilon )=\int\limits_{\mathcal{L}%
^{c}}\prod\limits_{i<j}P(\langle x_{i},x_{j}\rangle )^{\eta _{ij}}dx
\end{equation*}%
is the \textquotedblleft local contribution".

The global and local contributions are bounded in the following couple of
lemmas, whose proofs are given in Sections \ref{sec:glob contribution} and %
\ref{sec:loc contribution} respectively.

\begin{lemma}[``Global contribution"]
\bigskip \label{lem:glob contribution} As $l\rightarrow \infty $, $q\geq 5$
arbitrary, and $\underline{\eta} = \underline{\eta}(\gamma)$ with $\gamma\in
\Gamma _{C}(q,q,q,q)$, we have
\begin{eqnarray*}
|M_{glob}(\underline{\eta};\varepsilon)| =
\begin{cases}
O\left(\frac{(\log{l})^{2}}{l^{5}\varepsilon} \right) & q=5 \\
O\left(\frac{(\log{l})^{2}}{l^{6}\varepsilon^{2}}\right) & q=6 \\
O\left( \frac{\log{l}}{l^{q}\varepsilon^{q-3}} \right) & q\ge 7%
\end{cases}%
\end{eqnarray*}
with constant involved in the ``O"-notation depending on $q$ only.
\end{lemma}

\begin{lemma}[\textquotedblleft Local contribution"]
\label{lem:loc contribution} Under the assumptions of Lemma \ref{lem:glob
contribution} and the additional assumption
\begin{equation}
\varepsilon \gg \frac{1}{l},  \label{eq:epsilon>>1/l}
\end{equation}%
one has
\begin{equation}
|M_{loc}(\underline{\eta };\varepsilon )|=O\left( \frac{(\log {l}%
)^{3/2}\varepsilon ^{3/2}}{l^{3}}\right) .  \label{eq:Mloc<=}
\end{equation}
\end{lemma}

To prove lemmas \ref{lem:glob contribution} and \ref{lem:loc contribution} we will
use asymptotics for the $2$nd and $4$th moments of Legendre polynomials:
as $l\rightarrow\infty$
\begin{equation}
\label{eq:int(P^2)<<1/l}
\int\limits_{0}^{1}P_{l}(t)^{2}dt \sim c_{2}\frac{1}{l}.
\end{equation}
for some $c_{2}>0$ (cf. \eqref{legorth}), and
\begin{equation}
\label{eq:varhl4=4thmoment=logl/l^2}
\int\limits_{0}^{1}P_{l}(t)^{4}dt \sim c_{4}\frac{\log {l}}
{l^{2}}.
\end{equation}
for some $c_{4}>0$ (cf. Lemma \ref{lem:varhl4=4thmoment=logl/l^2}).

\subsubsection{Bounding the global contribution}

\label{sec:glob contribution}

\begin{proof}[Proof of Lemma \protect\ref{lem:glob contribution}]
By the definition of the domain of integration $\mathcal{L}(\varepsilon) $,
all the angles involved satisfy $\vartheta(x_{i},x_{j}) > \varepsilon$. Let $%
\underline{\eta}=\underline{\eta}(\gamma)$ corresponding to a connected
diagram $\gamma\in \Gamma _{C}(q,q,q,q)$.

First we assume that $q\ge 7$. By an easy counting argument,
\eqref{eq:sum
etaij=2q}, and Lemma \ref{lem:etaij=etai'j'}, it may be shown that
$\eta_{ij}\ge 2$ for (at least) two disjoint elements of $\underline{\eta}$.
With no loss of generality we assume that both $\eta_{12}\ge 2$ and also
$\eta_{34} \ge 2$. Since $\gamma$ is connected, we may also assume
$\eta_{13}\ge 1$, $\eta_{24}\ge 1$, again, by the virtue of Lemma
\ref{lem:etaij=etai'j'}.

Consider the integrand
\begin{equation}
K_{\underline{\eta }}(x_{1},\ldots ,x_{4})=\prod\limits_{i<j}P_{l}(\langle
x_{i},x_{j}\rangle )^{\eta _{ij}}  \label{eq:K integrand cum def}
\end{equation}%
of \eqref{eq:M glob def}. On $\mathcal{L}(\varepsilon )$ for every $i<j$ we
have the uniform upper bound
\begin{equation}  \label{eq:Pl(theta)<<sqrt(lepsilon) glob}
|P_{l}(\langle x_{i},x_{j}\rangle )|\ll \frac{1}{\sqrt{l\varepsilon }}
\end{equation}
by \eqref{eq:Legendre decay sqrt}. Hence
\begin{equation*}
|K_{\underline{\eta }}(x_{1},\ldots ,x_{4})|\ll \frac{1}{(l\varepsilon
)^{q-3}}P_{l}(\langle x_{1},x_{2}\rangle )^{2}P_{l}(\langle
x_{3},x_{4}\rangle )^{2}|P_{l}(\langle x_{1},x_{3}\rangle )P_{l}(\langle
x_{2},x_{4}\rangle )|,
\end{equation*}%
so that
\begin{equation}
|M_{glob}(\underline{\eta };\varepsilon )|\ll \frac{1}{(l\varepsilon )^{q-3}}%
\int\limits_{(\mathbb{S}^{2})^{4}}P_{l}(\langle x_{1},x_{2}\rangle
)^{2}P_{l}(\langle x_{3},x_{4}\rangle )^{2}|P_{l}(\langle x_{1},x_{3}\rangle
)P_{l}(\langle x_{2},x_{4}\rangle )|dx  \label{eq:Mglob bnd pointwise}
\end{equation}%
(note that at this stage we can afford to increase the domain of integration
to the whole of $(\mathbb{S}^{2})^{4}$).

In order to treat the latter integral we apply the Cauchy-Schwartz
inequality, dividing the $4$ terms into pairs. We team up each edge with its
disjoint complement (i.e. the unique edge with no common vertex); it is then
possible to split the variables in each of the resulting integrals. This
approach yields:
\begin{equation*}
\begin{split}
& \int\limits_{(\mathbb{S}^{2})^{4}}P_{l}(\langle x_{1},x_{2}\rangle
)^{2}P_{l}(\langle x_{3},x_{4}\rangle )^{2}|P_{l}(\langle x_{1},x_{3}\rangle
)P_{l}(\langle x_{2},x_{4}\rangle )|dx \\
& \leq \left( \int\limits_{(\mathbb{S}^{2})^{4}}P_{l}(\langle
x_{1},x_{2}\rangle )^{4}P_{l}(\langle x_{3},x_{4}\rangle )^{4}dx\right)
^{1/2}\cdot \left( \int\limits_{(\mathbb{S}^{2})^{4}}P_{l}(\langle
x_{1},x_{3}\rangle )^{2}P_{l}(\langle x_{2},x_{4}\rangle )^{2}dx\right)
^{1/2} \\
& \ll \frac{\log {l}}{l^{3}}.
\end{split}%
\end{equation*}%
by \eqref{eq:varhl4=4thmoment=logl/l^2} and \eqref{eq:int(P^2)<<1/l}. The
statement of the present lemma for $q\geq 6$ then follows upon substituting
the latter bound into \eqref{eq:Mglob bnd pointwise}.

For $q=5,6$ the same argument remains valid; however the bound it gives is
insufficient, and we will need to exploit the special structure of $%
\underline{\eta }$ in these cases. For $q=5$, $\underline{\eta }$ has one of
the following three shapes (in some order):
\begin{equation}
\underline{\eta }=(4,4,1,1,0,0)  \label{eq:eta=441100}
\end{equation}%
or
\begin{equation}
\underline{\eta }=(3,3,2,2,0,0)  \label{eq:eta=332200}
\end{equation}%
or
\begin{equation}
\underline{\eta }=(3,3,1,1,1,1),  \label{eq:eta=331111}
\end{equation}%
where, by Lemma \ref{lem:etaij=etai'j'}, the edges corresponding to $\eta
_{ij}=4,3,2$ are disjoint.

For the first case \eqref{eq:eta=441100}, we have (up to reordering $%
\{x_{i}\}$)
\begin{equation*}
M_{glob}(\underline{\eta };\varepsilon )=\int\limits_{\mathcal{L}%
(\varepsilon )}P_{l}(\langle x_{1},x_{2}\rangle )^{4}P_{l}(\langle
x_{3},x_{4}\rangle )^{4}P_{l}(\langle x_{1},x_{3}\rangle )P_{l}(\langle
x_{2},x_{4}\rangle )dx,
\end{equation*}%
so that the uniform bound \eqref{eq:Pl(theta)<<sqrt(lepsilon) glob} yields
\begin{equation*}
|M_{glob}(\underline{\eta };\varepsilon )|\ll \frac{1}{l\varepsilon }%
\int\limits_{\mathcal{L}(\varepsilon )}\int\limits_{(\mathbb{S}%
^{2})^{4}}P_{l}(\langle x_{1},x_{2}\rangle )^{4}P_{l}(\langle
x_{3},x_{4}\rangle )^{4}dx\ll \frac{(\log {l})^{2}}{\varepsilon l^{5}}
\end{equation*}%
by \eqref{eq:varhl4=4thmoment=logl/l^2}. Next, for the second case %
\eqref{eq:eta=332200}, we bound
\begin{equation*}
|M_{glob}(\underline{\eta };\varepsilon )|\ll \frac{1}{\varepsilon l}%
\int\limits_{(\mathbb{S}^{2})^{4}}P_{l}(\langle x_{1},x_{2}\rangle
)^{2}P_{l}(\langle x_{3},x_{4}\rangle )^{2}P_{l}(\langle x_{1},x_{3}\rangle
)^{2}P_{l}(\langle x_{2},x_{4}\rangle )^{2}dx,
\end{equation*}%
and continue as for $q\geq 7$ by bounding the latter integral using
Cauchy-Schwartz, pairing together disjoint edges.

Finally, for $\underline{\eta }$ as in \eqref{eq:eta=331111}, we use %
\eqref{eq:Pl(theta)<<sqrt(lepsilon) glob} to bound
\begin{equation*}
\begin{split}
& |M_{glob}(\underline{\eta };\varepsilon )| \\
& \ll \frac{1}{l\varepsilon }\int\limits_{(\mathbb{S}^{2})^{4}}P_{l}(\langle
x_{1},x_{2}\rangle )^{2}P_{l}(\langle x_{3},x_{4}\rangle )^{2}|P_{l}(\langle
x_{1},x_{3}\rangle )P_{l}(\langle x_{2},x_{4}\rangle )P_{l}(\langle
x_{1},x_{4}\rangle )P_{l}(\langle x_{2},x_{3}\rangle )|dx.
\end{split}%
\end{equation*}%
Pairing up the $4$th degree terms in the latter integral together and the
rest separately, we apply Cauchy-Schwartz as before, twice for the second
term. This leads to
\begin{equation*}
\begin{split}
& \int\limits_{(\mathbb{S}^{2})^{4}}P_{l}(\langle x_{1},x_{2}\rangle
)^{2}P_{l}(\langle x_{3},x_{4}\rangle )^{2}|P_{l}(\langle x_{1},x_{3}\rangle
)P_{l}(\langle x_{2},x_{4}\rangle )P_{l}(\langle x_{1},x_{4}\rangle
)P_{l}(\langle x_{2},x_{3}\rangle )|dx \\
& \leq \left( \int\limits_{(\mathbb{S}^{2})^{4}}P_{l}(\langle
x_{1},x_{2}\rangle )^{4}P_{l}(\langle x_{3},x_{4}\rangle )^{4}dx\right)
^{1/2}\cdot \left( \int\limits_{(\mathbb{S}^{2})^{4}}P_{l}(\langle
x_{1},x_{3}\rangle )^{4}P_{l}(\langle x_{2},x_{4}\rangle )^{4}dx\right)
^{1/4}\times \\
& \times \left( \int\limits_{(\mathbb{S}^{2})^{4}}P_{l}(\langle
x_{1},x_{4}\rangle )^{4}P_{l}(\langle x_{2},x_{3}\rangle )^{4}dx\right)
^{1/4};
\end{split}%
\end{equation*}%
this implies the statement of the present lemma for $q=5$ via separation of
variables and \eqref{eq:varhl4=4thmoment=logl/l^2}. The proof for
$q=6 $ is very similar to the above (but somewhat easier) and thereupon
omitted here. In this case $\eta $ is of one the following $5$ forms:
\begin{equation*}
\underline{\eta }%
=(5,5,1,1,0,0),(4,4,2,2,0,0),(4,4,1,1,1,1),(3,3,3,3,0,0),(3,3,2,2,1,1).
\end{equation*}
\end{proof}

\subsubsection{Bounding the local contribution}

\label{sec:loc contribution}

\begin{proof}[Proof of Lemma \protect\ref{lem:loc contribution}]
We may assume with no loss of generality that
\begin{equation*}
\vartheta(x_{1},x_{2}) < \varepsilon
\end{equation*}
in the relevant domain $\mathcal{L}(\varepsilon)^{c}$. We divide the
possibilities for $\underline{\eta}$ into three cases, up to reordering
(that $\underline{\eta}$ falls into one of those cases follows from %
\eqref{eq:sum etaij=2q} $q\ge 5$, Lemma \ref{lem:etaij=etai'j'} and the the
connectedness of $\gamma$):

\begin{enumerate}
\item \label{it:etaij>0} For every $i<j$, $\eta_{ij}>0$.

\item \label{it:441100} We have
\begin{equation*}
\eta_{12}=\eta_{34}\ge 4, \; \eta_{13}=\eta_{24} \ge 1,\;
\eta_{14}=\eta_{23}=0.
\end{equation*}

\item \label{it:332200} We have
\begin{equation*}
\eta_{12}=\eta_{34}\ge 2, \; \eta_{13}=\eta_{24} \ge 2.
\end{equation*}
\end{enumerate}

In case \ref{it:etaij>0} (which gives rise to the dominating term) we may
use Cauchy-Schwartz twice to bound
\begin{equation}
\begin{split}
& |M_{loc}(\underline{\eta };\varepsilon )|\leq \int\limits_{\mathcal{L}%
^{c}(\varepsilon )}\prod\limits_{i<j}|P(\langle x_{i},x_{j}\rangle )|dx \\
& \leq \left( \int\limits_{\mathcal{L}^{c}(\varepsilon )}P_{l}(\langle
x_{1},x_{2})^{2}P_{l}(\langle x_{3},x_{4})^{2}P_{l}(\langle
x_{1},x_{3})^{2}dx\right) ^{1/2}\times \\
& \times \left( \int\limits_{\mathcal{L}^{c}(\varepsilon )}P_{l}(\langle
x_{1},x_{4})^{2}P_{l}(\langle x_{2},x_{3})^{2}P_{l}(\langle
x_{2},x_{4})^{2}dx\right) ^{1/2}
\end{split}
\label{august}
\end{equation}
and bound each of the two integrals of \eqref{august} separately. For the
first integral we use Cauchy-Schwartz again:
\begin{equation*}
\begin{split}
& \int\limits_{\mathcal{L}^{c}(\varepsilon )}P_{l}(\langle
x_{1},x_{2})^{2}P_{l}(\langle x_{3},x_{4})^{2}P_{l}(\langle
x_{1},x_{3})^{2}dx \\
& \leq \left( \int\limits_{(\mathbb{S}^{2})^{4}}P_{l}(\langle
x_{1},x_{2})^{4}dx_{1}dx_{2}P_{l}(\langle
x_{3},x_{4})^{4}dx_{3}dx_{4}\right) ^{1/2}\cdot \left(
\int\limits_{\{x_{2}:\vartheta (x_{1},x_{2})<\varepsilon \}}P_{l}(\langle
x_{1},x_{3}\rangle )^{4}dx\right) ^{1/2} \\
& \ll \frac{\log {l}}{l^{2}}\cdot \frac{(\log {l})^{1/2}}{l}\varepsilon =%
\frac{(\log {l})^{3/2}}{l^{3}}\varepsilon ,
\end{split}%
\end{equation*}%
by separation of variables and
\begin{equation*}
\mu \left( \{x_{2}:\vartheta (x_{1},x_{2})<\varepsilon \}\right) \ll
\varepsilon ^{2}.
\end{equation*}

For the other integral of \eqref{august}, we may use the lack of symmetry
w.r.t. variables to improve the bound as follows:
\begin{equation*}
\begin{split}
& \int\limits_{\mathcal{L}^{c}(\varepsilon )}P_{l}(\langle
x_{1},x_{4})^{2}P_{l}(\langle x_{2},x_{3})^{2}P_{l}(\langle
x_{2},x_{4})^{2}dx \\
& \leq \left( \int\limits_{\mathcal{L}^{c}(\varepsilon )}P_{l}(\langle
x_{1},x_{4})^{4}dx_{1}dx_{4}\right) ^{1/2}\cdot \left(
\int\limits_{\{x_{2}:\vartheta (x_{1},x_{2})<\varepsilon \}}P_{l}(\langle
x_{2},x_{3})^{4}dx_{2}dx_{3}\right) ^{1/2}\times \\
& \times \left( \int\limits_{\{x_{2}:\vartheta (x_{1},x_{2})<\varepsilon
\}}P_{l}(\langle x_{2},x_{4})^{4}dx\right) ^{1/2}\ll \frac{(\log {l})^{1/2}}{%
l}\cdot \frac{(\log {l})^{1/2}}{l}\varepsilon \cdot \frac{(\log {l})^{1/2}}{l}%
\varepsilon \\
& =\frac{(\log {l})^{3/2}}{l^{3}}\varepsilon ^{2}.
\end{split}%
\end{equation*}%
We then obtain the statement of the present lemma in this case upon
substituting the last couple of estimates into \eqref{august}.

In case \ref{it:441100} we use similar ideas (Cauchy-Schwartz twice) to
obtain
\begin{equation*}
\begin{split}
& |M_{loc}(\underline{\eta };\varepsilon )|\leq \int\limits_{\mathcal{L}%
^{c}(\varepsilon )}P_{l}(\langle x_{1},x_{2}\rangle )^{4}P_{l}(\langle
x_{3},x_{4}\rangle )^{4}P_{l}(\langle x_{1},x_{3}\rangle )P_{l}(\langle
x_{2},x_{4}\rangle )dx \\
& \leq \left( \int\limits_{(\mathbb{S}^{2})^{4}}P_{l}(\langle
x_{1},x_{2}\rangle )^{4}P_{l}(\langle x_{3},x_{4}\rangle )^{4}P_{l}(\langle
x_{1},x_{3}\rangle )P_{l}(\langle x_{2},x_{4}\rangle )dx\right) ^{3/4}\times
\\
& \times \left( \int\limits_{\mathcal{L}^{c}(\varepsilon )}P_{l}(\langle
x_{1},x_{2}\rangle )^{4}P_{l}(\langle x_{3},x_{4}\rangle )^{4}P_{l}(\langle
x_{1},x_{3}\rangle )^{4}P_{l}(\langle x_{2},x_{4}\rangle )^{4}dx\right)
^{1/4} \\
& \ll \frac{(\log {l})^{3/2}}{l^{3}}\cdot \left( \int\limits_{(\mathbb{S}%
^{2})^{4}}P_{l}(\langle x_{1},x_{3}\rangle )^{4}dx_{1}dx_{3}\right)
^{1/4}\cdot \left( \int\limits_{\{x_{2}:\vartheta (x_{1},x_{2})<\varepsilon
\}}P_{l}(\langle x_{2},x_{4}\rangle )^{4}dx_{2}dx_{4}\right) ^{1/4} \\
& \ll \frac{(\log {l})^{2}}{l^{4}}\sqrt{\varepsilon },
\end{split}%
\end{equation*}
which is smaller than the RHS of \eqref{eq:Mloc<=} by \eqref{eq:epsilon>>1/l}%
.

Finally, for case \ref{it:332200} we similarly have
\begin{equation}
\begin{split}
& |M_{loc}(\underline{\eta };\varepsilon )|\leq \int\limits_{\mathcal{L}%
^{c}(\varepsilon )}\left\vert P_{l}(\langle x_{1},x_{2}\rangle
)^{3}\right\vert \left\vert P_{l}(\langle x_{3},x_{4}\rangle
)^{3}\right\vert P_{l}(\langle x_{1},x_{3}\rangle )^{2}P_{l}(\langle
x_{2},x_{4}\rangle )^{2}dx \\
& \leq \int\limits_{\mathcal{L}^{c}(\varepsilon )}P_{l}(\langle
x_{1},x_{2}\rangle )^{2}P_{l}(\langle x_{3},x_{4}\rangle )^{2}P_{l}(\langle
x_{1},x_{3}\rangle )^{2}P_{l}(\langle x_{2},x_{4}\rangle )^{2}dx \\
& \leq \left( \int\limits_{(\mathbb{S}^{2})^{2}}P_{l}(\langle
x_{1},x_{2}\rangle )^{4}dx_{1}dx_{2}\right) ^{1/2}\cdot \left( \int\limits_{(%
\mathbb{S}^{2})^{2}}P_{l}(\langle x_{3},x_{4}\rangle
)^{4}dx_{3}dx_{4}\right) ^{1/2} \\
& \left( \int\limits_{(\mathbb{S}^{2})^{2}}P_{l}(\langle x_{1},x_{3}\rangle
)^{4}dx_{1}dx_{3}\right) ^{1/2}\cdot \left( \int\limits_{\{x_{2}:\vartheta
(x_{1},x_{2})<\varepsilon \}}P_{l}(\langle x_{2},x_{4}\rangle
)^{4}dx_{2}dx_{4}\right) ^{1/2} \\
& \ll \frac{\log {l}}{l^{2}}\cdot \frac{(\log {l})^{1/2}}{l}\cdot \frac{%
(\log {l})^{1/2}}{l}\varepsilon =\frac{(\log {l})^{2}}{l^{4}}\varepsilon ,
\end{split}
\label{eq:Mloc111111<<CS}
\end{equation}%
which is less than latter of the expressions on the RHS of \eqref{eq:Mloc<=}%
, again by \eqref{eq:epsilon>>1/l}.
\end{proof}

\subsubsection{Concluding the proof of Proposition \protect\ref{prop:cums
are small}}

\label{sec:concl proof}

\begin{proof}
In order to finish the proof of the present proposition it remains to make a
suitable choice of the parameter $\varepsilon (l)$ so that both the local
and the global contributions as bounded by Lemmas \ref{lem:glob contribution}
and \ref{lem:loc contribution} will be smaller than the expressions
prescribed in Proposition \ref{prop:cums are small}. The optimal choice for
the arising trade-off is
\begin{equation*}
\varepsilon (l)=
\begin{cases}
\frac{1}{l^{4/5}} & q=5 \\
\frac{1}{l^{6/7}} & q=6 \\
\frac{1}{l^{1-\frac{3}{2q-3}}} & q\geq 7%
\end{cases}%
,
\end{equation*}%
giving the bound in the statement of the present proposition.
\end{proof}

\section{The Central Limit Theorem for $h_{l;q}$, $q=3,4.$ \label{q=3,4}}

We shall start from the investigation of total variation bounds for $%
h_{l;3}; $ this result was established in \cite{MaPeCUP}, see also \cite%
{M2008}, \cite{MaWi}, but nevertheless we report it here for the sake of
completeness.
The Lemmas below make some use of so-called Wigner's and Clebsch-Gordan
coefficients; see the Appendix for their definition and discussion
of some important properties.

\begin{lemma}
\label{q=3}

1. The variance of $h_{l;3}$ is given by%
\begin{equation}
\mathbb{E}\left[ h_{l;3}^{2}\right] =6\times (4\pi )^{2}\left(
\begin{array}{ccc}
l & l & l \\
0 & 0 & 0%
\end{array}%
\right) ^{2}\sim \frac{12}{\pi \sqrt{3}}\frac{(4\pi )^{2}}{l^{2}}\text{ ;}
\label{var3}
\end{equation}%
2. For the fourth-order cumulant of $h_{l;3}$ we have%
\begin{equation}
cum_{4}(h_{l;3})\sim \frac{1}{2l+1}\left(
\begin{array}{ccc}
l & l & l \\
0 & 0 & 0%
\end{array}%
\right) ^{4}\text{ ;}  \label{cum4}
\end{equation}
3. The following total variation bound holds:
\begin{equation}
d_{TV}\left(\frac{h_{l;3}}{\sqrt{Var(h_{l;3})}},\mathcal{N}(0,1)\right)=O\left(\frac{1}{%
\sqrt{l}}\right).  \label{boundq3}
\end{equation}
\end{lemma}

As argued in the Introduction, the general strategy for the proofs of our
convergence results requires a careful evaluation of the variance and
suitable bounds on fourth-order cumulants. For $q=4,$ these computations are
provided in the two Lemmas to follow.

\begin{lemma}
\label{lem:varhl4=4thmoment=logl/l^2}

\begin{enumerate}
\item The variance of $h_{l;4}$ is given by%
\begin{equation*}
Var\left\{ h_{l;4}\right\} =4!(4\pi )^{2}\int\limits_{0}^{1}P_{l}(t)^{4}dt
\end{equation*}

\item As $l\rightarrow \infty $ we have
\begin{equation*}
\int\limits_{0}^{1}P_{l}(t)^{4}dt\sim \frac{3}{2\pi ^{2}}\frac{\log {l}}{%
l^{2}}.
\end{equation*}%
In particular,
\begin{equation*}
Var\left\{ h_{l;4}\right\} \sim 24^{2}\frac{\log {l}}{l^{2}}\text{ .}
\end{equation*}
\end{enumerate}
\end{lemma}

\begin{proof}
Since
\begin{equation*}
\mathbb{E}[h_{l;4}]=0,
\end{equation*}%
the first part of the present lemma follows from
\begin{equation*}
\begin{split}
& Var\left\{ h_{l;4}\right\} =\mathbb{E}[h_{l;4}^{2}]=\mathbb{E}\left[
\int\limits_{\mathbb{S}^{2}}H_{4}(f_{l}(x))dx\cdot \int\limits_{\mathbb{S}%
^{2}}H_{4}(f_{l}(y))dy\right] \\
& =\int\limits_{\mathbb{S}^{2}\times \mathbb{S}^{2}}\mathbb{E}\left[
H_{4}(f_{l}(x))H_{4}(f_{l}(y))\right] dxdy=4!\int\limits_{\mathbb{S}%
^{2}\times \mathbb{S}^{2}}P_{l}(\langle x,\,y\rangle )^{4}dxdy=4!(4\pi
)^{2}\int\limits_{0}^{1}P_{l}(t)^{4}dt.
\end{split}%
\end{equation*}

To see the second part, we invoke Hilb's asymptotics
\eqref{eq:Hilb
asymptotics} in the Appendix to write (up to an admissible error, as it is
easy to directly check)
\begin{equation*}
\begin{split}
\int\limits_{0}^{1}P_{l}(t)^{4}dt& \sim \frac{1}{l}\int\limits_{1}^{l+1/2}%
\sin \left( \frac{\psi }{l+1/2}\right) J_{0}(\psi )^{4}d\psi \sim \frac{1}{l}%
\int\limits_{1}^{l}\frac{\psi }{l}\cdot J_{0}(\psi )^{4}d\psi \\
& \sim \frac{1}{l^{2}}\int\limits_{1}^{l}\psi \cdot \frac{4}{\pi ^{2}}\frac{%
\sin (\psi +\pi /4)^{4}}{\psi ^{2}}d\psi \sim \frac{1}{l^{2}}\cdot \frac{3}{%
2\pi ^{2}}\int\limits_{1}^{l}\frac{d\psi }{\psi},
\end{split}%
\end{equation*}
by the standard asymptotics for the Bessel $J_{0}$.
\end{proof}

\begin{lemma}
\label{lem:cum bnd q=4} As $l\rightarrow \infty ,$ we have%
\begin{equation*}
cum_{4}\left\{ h_{l;4}\right\} \approx l^{-4}\text{ .}
\end{equation*}%
\bigskip
\end{lemma}

\begin{proof}
For our purposes, we need to show that
\begin{equation}
A_{1}:=\int_{S^{2}\times ....S^{2}}P_{l}(\left\langle w,z\right\rangle
)P_{l}^{3}(\left\langle w,w^{\prime }\right\rangle )P_{l}(\left\langle
w^{\prime },z^{\prime }\right\rangle )P_{l}^{3}(\left\langle z^{\prime
},z\right\rangle )dwdzdw^{\prime }dz^{\prime }=O\left( \frac{\log ^{2}l}{%
l^{5}}\right) \text{ }  \label{anotte1}
\end{equation}%
and%
\begin{equation}
A_{2}:=\int_{S^{2}\times ....S^{2}}P_{l}^{2}(\left\langle w,z\right\rangle
)P_{l}^{2}(\left\langle w,w^{\prime }\right\rangle )P_{l}^{2}(\left\langle
w^{\prime },z^{\prime }\right\rangle )P_{l}^{2}(\left\langle z^{\prime
},z\right\rangle )dwdzdw^{\prime }dz^{\prime }\approx \frac{1}{l^{4}}\text{ .%
}  \label{bnotte1}
\end{equation}%
Concerning the first term, we note that%
\begin{equation*}
\int_{\mathbb{S}^{2}}P_{l}(\left\langle w,z\right\rangle
)P_{l}^{3}(\left\langle w,w^{\prime }\right\rangle )dw
\end{equation*}%
\begin{equation*}
=\left\{ \frac{4\pi }{2l+1}\right\} ^{4}\int_{\mathbb{S}^{2}}\left\{
\sum_{m}Y_{lm}(w)\overline{Y}_{lm}(z)\right\} \left\{ \sum_{m^{\prime
}}Y_{lm^{\prime }}(w)\overline{Y}_{lm^{\prime }}(w^{\prime })\right\} ^{3}dw
\end{equation*}%
\begin{eqnarray*}
&=&\left\{ \frac{4\pi }{2l+1}\right\} ^{4}\sum_{m_{1}m_{2}^{\prime
}m_{3}^{\prime }m_{4}^{\prime }}\overline{Y}_{lm_{1}}(z)\overline{Y}%
_{lm_{2}^{\prime }}(w^{\prime })\overline{Y}_{lm_{3}^{\prime }}(w^{\prime })%
\overline{Y}_{lm_{4}^{\prime }}(w^{\prime }) \\
&&\times \int_{\mathbb{S}^{2}}Y_{lm_{1}}(w)Y_{lm_{2}^{\prime
}}(w)Y_{lm_{3}^{\prime }}(w)Y_{lm_{4}^{\prime }}(w^{\prime })dw
\end{eqnarray*}%
\begin{equation*}
=\frac{(4\pi )^{3}}{(2l+1)^{2}}\sum_{m_{1}m_{2}^{\prime }m_{3}^{\prime
}m_{4}^{\prime }}\overline{Y}_{lm_{1}}(z)\overline{Y}_{lm_{2}^{\prime
}}(w^{\prime })\overline{Y}_{lm_{3}^{\prime }}(w^{\prime })\overline{Y}%
_{lm_{4}^{\prime }}(w^{\prime })
\end{equation*}%
\begin{equation*}
\times \sum_{LM}(-1)^{L+M}\left\{ C_{l0l0}^{L0}\right\} ^{2}\frac{%
C_{lm_{1}lm_{2}^{\prime }}^{LM}C_{m_{3}^{\prime }m_{4}^{\prime }}^{L,-M}}{%
2L+1}\ .
\end{equation*}%
Likewise%
\begin{equation*}
\int_{\mathbb{S}^{2}}P_{l}(\left\langle w^{\prime },z^{\prime }\right\rangle
)P_{l}^{3}(\left\langle z^{\prime },z\right\rangle )dz^{\prime }=
\end{equation*}%
\begin{eqnarray*}
&=&\frac{(4\pi )^{3}}{(2l+1)^{2}}\sum_{m_{1}^{\prime }m_{2}m_{3}m_{4}}%
\overline{Y}_{lm_{1}^{\prime }}(w^{\prime })\overline{Y}_{lm_{2}}(z^{\prime
})\overline{Y}_{lm_{4}}(z^{\prime })\overline{Y}_{lm_{4}}(z^{\prime }) \\
&&\times \sum_{L^{\prime }M^{\prime }}(-1)^{L+M}\left\{ C_{l0l0}^{L^{\prime
}0}\right\} ^{2}\frac{C_{lm_{1}^{\prime }lm_{2}}^{L^{\prime }M^{\prime
}}C_{m_{3}m_{4}}^{L^{\prime },-M^{\prime }}}{2L^{\prime }+1}\ .
\end{eqnarray*}%
Note that $L$ is necessarily even here, otherwise the Clebsch-Gordan
coefficients $\left\{ C_{l0l0}^{L0}\right\} $ are identically null from (\ref%
{appe}), whence $(-1)^{L+M}=(-1)^{M}.$ Iterating the same argument twice
more we obtain%
\begin{equation*}
A_{1}\approx \sum_{m_{1}...m_{4}^{\prime \prime }}\sum_{LM}(-1)^{M}\left\{
C_{l0l0}^{L0}\right\} ^{2}\frac{C_{lm_{1}lm_{2}^{\prime
}}^{LM}C_{m_{3}^{\prime }m_{4}^{\prime }}^{L,-M}}{2L+1}\sum_{L^{\prime
}M^{\prime }}(-1)^{M^{\prime }}\left\{ C_{l0l0}^{L^{\prime }0}\right\} ^{2}%
\frac{C_{lm_{1}^{\prime }lm_{2}^{\prime }}^{L^{\prime }M^{\prime
}}C_{m_{3}^{\prime }m_{4}^{\prime }}^{L^{\prime },-M^{\prime }}}{2L^{\prime
}+1}
\end{equation*}%
\begin{equation}
\times \sum_{L^{\prime \prime }M^{\prime \prime }}(-1)^{M^{\prime \prime
}}\left\{ C_{l0l0}^{L^{\prime \prime }0}\right\} ^{2}\frac{C_{lm_{1}^{\prime
}lm_{2}^{\prime \prime }}^{L^{\prime \prime }M^{\prime \prime
}}C_{m_{3}^{\prime \prime }m_{4}^{\prime \prime }}^{L^{\prime \prime
},-M^{\prime \prime }}}{2L^{\prime \prime }+1}\sum_{L^{\prime \prime \prime
}M^{\prime \prime \prime }}(-1)^{M^{\prime \prime \prime }}\left\{
C_{l0l0}^{L^{\prime \prime \prime }0}\right\} ^{2}\frac{C_{lm_{1}lm_{2}^{%
\prime \prime }}^{L^{\prime \prime \prime }M^{\prime \prime
}}C_{m_{3}^{\prime \prime }m_{4}^{\prime \prime }}^{L^{\prime \prime \prime
},-M^{\prime \prime }}}{2L^{\prime \prime \prime }+1}\text{ .}
\label{terma1}
\end{equation}%
Now, applying iteratively the orthogonality identity (\ref{ortho1}), we have%
\begin{equation*}
\sum_{m_{2}^{\prime }m_{3}^{\prime }m_{4}^{\prime }}\sum_{LM}\sum_{L^{\prime
}M^{\prime }}(-1)^{M+M^{\prime }}\left\{ C_{l0l0}^{L0}\right\} ^{2}\frac{%
C_{lm_{1}lm_{2}^{\prime }}^{LM}C_{m_{3}^{\prime }m_{4}^{\prime }}^{L,-M}}{%
2L+1}\left\{ C_{l0l0}^{L^{\prime }0}\right\} ^{2}\frac{C_{lm_{1}^{\prime
}lm_{2}^{\prime }}^{L^{\prime }M^{\prime }}C_{m_{3}^{\prime }m_{4}^{\prime
}}^{L^{\prime },-M^{\prime }}}{2L^{\prime }+1}
\end{equation*}%
\begin{equation*}
=\sum_{m_{2}^{\prime }}\sum_{LM}\frac{\left\{ C_{l0l0}^{L0}\right\} ^{4}}{%
2L+1}\frac{C_{lm_{1}lm_{2}^{\prime }}^{LM}C_{m_{1}^{\prime }m_{2}^{\prime
}}^{LM}}{2L+1}=\sum_{L}\frac{\left\{ C_{l0l0}^{L0}\right\} ^{4}}{(2l+1)(2L+1)%
}\delta _{m_{1}}^{m_{1}^{\prime }}\text{ .}
\end{equation*}%
Applying the same argument to the last two terms in (\ref{terma1}), we
obtain that%
\begin{equation*}
A_{1}\approx \sum_{m_{1},m_{1}^{\prime },L,L^{\prime }}\frac{\left\{
C_{l0l0}^{L0}\right\} ^{4}}{2L+1}\frac{\left\{ C_{l0l0}^{L^{\prime
}0}\right\} ^{4}}{2L^{\prime }+1}\frac{\delta _{m_{1}}^{m_{1}^{\prime }}}{%
(2l+1)^{2}}
\end{equation*}%
\begin{equation*}
=\frac{1}{2l+1}\left\{ \sum_{L}\frac{\left\{ C_{l0l0}^{L0}\right\} ^{4}}{2L+1%
}\right\} ^{2}=O\left( \frac{1}{2l+1}\frac{\log ^{2}l}{l^{4}}\right) \text{ .%
}
\end{equation*}%
Here it is interesting to recall that
\begin{equation*}
\sum_{L}\frac{\left\{ C_{l0l0}^{L0}\right\} ^{4}}{2L+1}=\sum_{L}(2L+1)\left(
\begin{array}{ccc}
l & l & L \\
0 & 0 & 0%
\end{array}%
\right) ^{4}=\int_{0}^{1}P_{l}^{4}(t)dt=\frac{Var\left\{ h_{l;4}\right\} }{%
4!(4\pi )^{2}}\text{ .}
\end{equation*}%
see also \cite{MaWi}, Lemma A1 and the proof of Lemma 2.3 therein. \

Let us now focus on (\ref{bnotte1}). Using again (\ref{exploiting}) and (\ref%
{gaunt4}), we have that%
\begin{equation*}
\int_{\mathbb{S}^{2}}P_{l}^{2}(\left\langle w,z\right\rangle
)P_{l}^{2}(\left\langle w,w^{\prime }\right\rangle )dw
\end{equation*}%
\begin{equation*}
=\left\{ \frac{4\pi }{2l+1}\right\} ^{4}\int_{\mathbb{S}^{2}}\left\{
\sum_{m}Y_{lm}(w)\overline{Y}_{lm}(z)\right\} ^{2}\left\{ \sum_{m^{\prime
}}Y_{lm^{\prime }}(w)\overline{Y}_{lm^{\prime }}(w^{\prime })\right\} ^{2}dw
\end{equation*}%
\begin{eqnarray*}
&=&\left\{ \frac{4\pi }{2l+1}\right\} ^{4}\sum_{m_{1}m_{2}m_{3}^{\prime
}m_{4}^{\prime }}\overline{Y}_{lm_{1}}(z)\overline{Y}_{lm_{2}}(z)\overline{Y}%
_{lm_{3}^{\prime }}(w^{\prime })\overline{Y}_{lm_{4}^{\prime }}(w^{\prime })
\\
&&\times \int_{\mathbb{S}^{2}}Y_{lm_{1}}(w)Y_{lm_{2}}(w)Y_{lm_{3}^{\prime
}}(w)Y_{lm_{4}^{\prime }}(w^{\prime })dw
\end{eqnarray*}%
\begin{equation*}
=\frac{(4\pi )^{3}}{(2l+1)^{2}}\sum_{m_{1}m_{2}m_{3}^{\prime }m_{4}^{\prime
}}\overline{Y}_{lm_{1}}(z)\overline{Y}_{lm_{2}}(z)\overline{Y}%
_{lm_{3}^{\prime }}(w^{\prime })\overline{Y}_{lm_{4}^{\prime }}(w^{\prime })
\end{equation*}%
\begin{equation*}
\times \sum_{LM}\left\{ C_{l0l0}^{L0}\right\} ^{2}\frac{%
C_{lm_{1}lm_{2}}^{LM}C_{m_{3}^{\prime }m_{4}^{\prime }}^{L,-M}}{2L+1}\ .
\end{equation*}%
Iterating the argument, we find that
\begin{equation*}
A_{2}\approx \sum_{m_{1}...m_{4}^{\prime \prime }}\sum_{LM}(-1)^{M}\left\{
C_{l0l0}^{L0}\right\} ^{2}\frac{C_{lm_{1}lm_{2}}^{LM}C_{m_{3}^{\prime
}m_{4}^{\prime }}^{L,-M}}{2L+1}\sum_{L^{\prime }M^{\prime }}(-1)^{M^{\prime
}}\left\{ C_{l0l0}^{L^{\prime }0}\right\} ^{2}\frac{C_{lm_{1}^{\prime
}lm_{2}^{\prime }}^{L^{\prime }M^{\prime }}C_{m_{3}^{\prime }m_{4}^{\prime
}}^{L^{\prime },-M^{\prime }}}{2L+1}
\end{equation*}%
\begin{equation}
\times \sum_{L^{\prime \prime }M^{\prime \prime }}(-1)^{M^{\prime \prime
}}\left\{ C_{l0l0}^{L^{\prime \prime }0}\right\} ^{2}\frac{C_{lm_{1}^{\prime
}lm_{2}^{\prime }}^{L^{\prime \prime }M^{\prime \prime }}C_{m_{3}^{\prime
\prime }m_{4}^{\prime \prime }}^{L^{\prime \prime },-M^{\prime \prime }}}{%
2L+1}\sum_{L^{\prime \prime \prime }M^{\prime \prime \prime
}}(-1)^{M^{\prime \prime \prime }}\left\{ C_{l0l0}^{L^{\prime \prime \prime
}0}\right\} ^{2}\frac{C_{lm_{1}lm_{2}}^{L^{\prime \prime }M^{\prime \prime
\prime }}C_{m_{3}^{\prime \prime }m_{4}^{\prime \prime }}^{L^{\prime \prime
},-M^{\prime \prime \prime }}}{2L+1}\text{ .}  \label{termb1}
\end{equation}%
Again it is sufficient to apply (\ref{ortho1}) four times to have%
\begin{equation*}
\sum_{L,L^{\prime },L^{\prime \prime },L^{\prime \prime \prime
}}\sum_{M,M^{\prime },M^{\prime \prime },M^{\prime \prime \prime }}\left\{
C_{l0l0}^{L0}\right\} ^{2}\left\{ C_{l0l0}^{L^{\prime }0}\right\} ^{2}
\end{equation*}%
\begin{equation*}
\times \left\{ C_{l0l0}^{L^{\prime \prime }0}\right\} ^{2}\left\{
C_{l0l0}^{L^{\prime \prime \prime }0}\right\} ^{2}\frac{\delta
_{L}^{L^{\prime }}\delta _{L^{\prime }}^{L^{\prime \prime }}\delta
_{L^{\prime \prime }}^{L^{\prime \prime \prime }}\delta _{L^{\prime \prime
\prime }}^{L}}{(2L+1)(2L^{\prime }+1)}\frac{\delta _{M}^{M^{\prime }}\delta
_{M^{\prime }}^{M^{\prime \prime }}\delta _{M^{\prime \prime }}^{M^{\prime
\prime \prime }}\delta _{M^{\prime \prime \prime }}^{M}}{(2L^{\prime \prime
}+1)(2L^{\prime \prime \prime }+1)}
\end{equation*}%
\begin{equation*}
=\sum_{LM}\left(
\begin{array}{ccc}
l & l & L \\
0 & 0 & 0%
\end{array}%
\right) ^{8}=\sum_{L}(2L+1)\left(
\begin{array}{ccc}
l & l & L \\
0 & 0 & 0%
\end{array}%
\right) ^{8}\text{ }.
\end{equation*}%
We shall now need the following result (see again \cite{MaWi}, Lemma A.1):
\begin{equation*}
\left(
\begin{array}{ccc}
l & l & L \\
0 & 0 & 0%
\end{array}%
\right) ^{2}=\gamma _{lL}\times \frac{2}{\pi }\times \frac{1}{%
L(2l-L)^{1/2}(2l+L)^{1/2}}\text{ , }\frac{1}{2}\leq \gamma _{lL}\leq \frac{8%
}{5}\text{ .}
\end{equation*}%
Simple manipulations then yield%
\begin{equation*}
\sum_{L}(2L+1)\left(
\begin{array}{ccc}
l & l & L \\
0 & 0 & 0%
\end{array}%
\right) ^{8}\leq \sum_{L=0}^{2l-2}\frac{(2L+1)}{L^{4}(2l-L)^{2}(2l+L)^{2}}
\end{equation*}%
\begin{eqnarray*}
&\leq &\sum_{L=0}^{l}\frac{(2L+1)}{L^{4}(2l-L)^{2}(2l+L)^{2}}%
+\sum_{L=l}^{2l-2}\frac{(2L+1)}{L^{4}(2l-L)^{2}(2l+L)^{2}} \\
&\leq &\frac{1}{4l^{4}}\sum_{L=0}^{l}\frac{(2L+1)}{L^{4}}+\frac{1}{l^{4}}%
\sum_{L=l}^{2l-2}\frac{(2L+1)}{L(2l-L)^{2}(2l+L)} \\
&=&O(l^{-4})+\frac{1}{l^{5}}\sum_{L=l}^{2l-2}\frac{1}{(2l-L)^{2}}=O(l^{-4})%
\text{ ,}
\end{eqnarray*}%
which completes the proof of the upper bound. To prove that this bound is
sharp, it suffices to notice that%
\begin{equation*}
\sum_{L}(2L+1)\left(
\begin{array}{ccc}
l & l & L \\
0 & 0 & 0%
\end{array}%
\right) ^{8}\geq \left(
\begin{array}{ccc}
l & l & 0 \\
0 & 0 & 0%
\end{array}%
\right) ^{8}=\frac{1}{(2l+1)^{4}}\text{ .}
\end{equation*}
\end{proof}

Combining the variance and cumulant results, and exploiting (\ref{noupebou2}%
), one finally obtains the following result.

\begin{proposition}
As $l\rightarrow \infty ,$ we have
\begin{equation*}
d_{TV}\left( \frac{h_{4;l}}{\sqrt{Var(h_{4;l})}},\mathcal{N}(0,1)\right)
=O\left( \frac{1}{\log l}\right) \text{ .}
\end{equation*}
\end{proposition}

\section{Applications \label{applications}}

\subsection{Polyspectra for spherical random fields}

Let $T(x)$ be a zero-mean Gaussian and isotropic spherical random field,
i.e. a measurable application $T:\mathcal{S}^{2}\times \Omega \rightarrow
\mathbb{R}$ such that $T(x)\overset{d}{=}T(gx)$ for all elements of the
group of rotations $g\in SO(3).$ It is well-known that the following
mean-square representation holds, in the $L^{2}(dx\times dP)$ sense (see
\cite{MaPeCUP}, Chapter 5):
\begin{equation*}
T(x)=\sum_{l}T_{l}(x)\text{ , where }\Delta _{\mathbb{S}%
^{2}}T_{l}=-l(l+1)T_{l}\text{ .}
\end{equation*}%
We can hence view the eigenfunctions $f_{l}$ as the normalized Fourier
components of such spherical field, e.g. $f_{l}(x):=T_{l}(x)/\sqrt{\mathbb{E}%
\left[ T_{l}^{2}(x)\right] }.$ In this subsection, we shall consider the
central limit theorem for polynomial functionals of the form
\begin{equation*}
Z_{l}=\sum_{q=0}^{Q}b_{q}\int_{S^{2}}\left\{ f_{l}(x)\right\} ^{q}dx\text{ ,
for some }Q\in \mathbb{N}\text{ , }b_{q}\in \mathbb{R}\text{ .}
\end{equation*}%
When we view the eigenfunctions $f_{l}$ as the Fourier components of an
isotropic spherical random field, these polynomial statistics cover, for
instance, the well-known (moment and cumulant)\emph{\ polyspectra} of the
random field. These are the crucial statistics when searching, for instance,
for possible non-Gaussian behaviour in $T(x);$ see for instance \cite%
{bartolo}, \cite{lewis}, and the references therein. Note that there exist
deterministic coefficients $\beta _{0},...,\beta _{p}$ such that we can
write
\begin{equation*}
Z_{l}=\sum_{q=0}^{Q}\beta
_{q}\int_{S^{2}}H_{q}(f_{2l}(x))dx=\sum_{q=0}^{Q}\beta _{q}h_{2l;q}\text{ .}
\end{equation*}%
From the results in the previous Section, we have immediately the following

\begin{corollary}
Assume that $c_{q}>0$ for at least one $q$ such that $\beta _{q}\neq 0.$
Then
\begin{equation*}
\frac{Z_{l}-\mathbb{E}\left[ Z_{l}\right] }{\sqrt{Var(Z_{l})}}\rightarrow
_{d}\mathcal{N}(0,1)\text{ , as }l\rightarrow \infty \text{ .}
\end{equation*}
\end{corollary}

The proof is immediate in light of Theorem \ref{thm:hlq asymp gauss}. Indeed, we are
dealing here with a finite linear combination of asymptotically Gaussian
random variables, and we recall that for random vectors with components in
Wiener chaoses, the multivariate Central Limit Theorem follows from
convergence in distribution of the univariate components, see \cite%
{peccatitudor}. This result thus extends the Central Limit Theorem provided
in \cite{M2008} for the sequence $\left\{ h_{2l;3}\right\} $ to polyspectra
of arbitrary orders.

It is actually possible to establish stronger results, i.e. to study the
rates of convergence in the total variation bound. Rather than focusing on
this issue, we move to the Central Limit Theorem for the case of a more
general, infinite-order $L^{2}$ expansion, as it is the case for the
\emph{Defect}.

\subsection{Defect}

\label{sec:defect}

In this subsection, we shall focus on one of the most important geometric
functionals, namely the Defect. The Defect (or \textquotedblleft
signed area\textquotedblright , see \cite{BGS}) of a function $\psi :\mathbb{%
\ S}^{2}\rightarrow \mathbb{R}$ is defined as
\begin{equation}
\mathcal{D}(\psi ):=\mathrm{meas}\left( \psi ^{-1}(0,\infty )\right) -%
\mathrm{\ meas}\left( \psi ^{-1}(-\infty ,0)\right) =\int_{\mathbb{S}^{2}}%
\mathcal{H}(\psi (x))dx.  \label{eq:defect def}
\end{equation}%
Here $\mathcal{H}(t)$ is such that
\begin{equation}
\mathcal{H}(t)=\mathds{1}_{[0,\infty )}(t)-\mathds{1}_{(-\infty ,0]}(t)=%
\begin{cases}
1 & t>0 \\
-1 & t<0 \\
0 & t=0%
\end{cases}%
,  \label{eq:heaviside def}
\end{equation}%
where $\mathds{1}_{A}(t)$ is the usual indicator function of the set $A$,
and $dx$ is the Lebesgue measure. In our case, the Defect is the difference
between the areas of positive and negative inverse image of $f_{l} $, denoted
\begin{equation*}
\mathcal{D}_{l}:=\mathcal{D}(f_{l}).
\end{equation*}

It has been shown by \cite{MaWi} that the following expansion holds, in the $%
L^{2}(dP)$ sense%
\begin{equation*}
\mathcal{D}_{l}=\sum_{q=1}^{\infty }\frac{J_{2q+1}}{(2q+1)!}%
h_{l;2q+1}=\sum_{q=1}^{\infty }\frac{(-1)}{\sqrt{2\pi }}\frac{(2q-1)!!}{%
(2q+1)!}h_{l;2q+1}\text{ .}
\end{equation*}%
Trivially $\mathbb{E}\left[ \mathcal{D}_{l}\right] =0$, and from \cite%
{MaWi2} we have that%
\begin{equation*}
Var(\mathcal{D}_{l})=\mathbb{E}\left[ \mathcal{D}_{l}^{2}\right] \sim
\sum_{q=1}^{\infty }a_{q}\frac{c_{2q+1}}{l^{2}}+o(l^{-2})\text{ , }a_{q}=%
\frac{(2q)!}{4^{q}(q!)^{2}(2q+1)}
\end{equation*}
are the (suitably normalized) Taylor coefficients of $\arcsin$ are asymptotic
to
\begin{equation*}
a_{q}=\frac{1}{2\sqrt{\pi }q^{3/2}}+o(q^{-3/2})\text{ , as }q\rightarrow
\infty
\end{equation*}%
by Stirling's formula, and
\begin{equation*}
\sum_{q=1}^{\infty }a_{q}c_{2q+1}>\frac{32}{\sqrt{27}}\text{ .}
\end{equation*}%
Note that we know $c_{3}>0$ from (\ref{var3}); any term corresponding to $%
c_{2q+1}=0$ can simply be dropped from the expansion, so the rate for this
variance is precise. In view of Theorem \ref{thm:hlq asymp gauss}, it is then
not difficult to prove the following result.

\begin{corollary}
As $l\rightarrow \infty$, we have
\begin{equation*}
\frac{\mathcal{D}_{2l}}{\sqrt{Var(\mathcal{D}_{2l})}}\rightarrow _{d}%
\mathcal{N}(0,1)\text{.}
\end{equation*}
\end{corollary}

\begin{proof}
The proof follows a standard argument for nonlinear transforms of Gaussian
measures, see for instance \cite{Sodin&Tsirelson}. Define
\begin{equation*}
\mathcal{D}_{l;m}:=\sum_{q=1}^{m}\frac{J_{2q+1}}{(2q+1)!}h_{l;2q+1}\text{ ;}
\end{equation*}%
using the trivial inequality $\mathbb{E}\left[ (A-B)^{2}\right] \leq 2%
\mathbb{E}\left[ (A-C)^{2}\right] +2\mathbb{E}\left[ (C-B)^{2}\right] ,$ we
have that%
\begin{eqnarray*}
&&\mathbb{E}\left[ \left( \frac{\mathcal{D}_{l}}{\sqrt{Var(\mathcal{D}_{l})}}%
-\frac{\mathcal{D}_{l;m}}{\sqrt{Var(\mathcal{D}_{l;m})}}\right) ^{2}\right]
\\
&\leq &2\mathbb{E}\left[ \left( \frac{\mathcal{D}_{l}}{\sqrt{Var(\mathcal{D}%
_{l})}}-\frac{\mathcal{D}_{l;m}}{\sqrt{Var(\mathcal{D}_{l})}}\right) ^{2}%
\right] +2\mathbb{\ E}\left[ \left( \frac{\mathcal{D}_{l;m}}{\sqrt{Var(%
\mathcal{D}_{l})}}-\frac{\mathcal{D}_{l;m}}{\sqrt{Var(\mathcal{D}_{l;m})}}%
\right) ^{2}\right] \\
&\leq &\frac{2}{Var(\mathcal{D}_{l})}\mathbb{E}\left[ (\mathcal{D}_{l}-%
\mathcal{D}_{l;m})^{2}\right] +2\left( \frac{Var(\mathcal{D}_{l;m})}{Var(%
\mathcal{D}_{l})}+1-\frac{2\sqrt{Var(\mathcal{D}_{l;m})}}{\sqrt{Var(\mathcal{%
D}_{l})}}\right) \text{ .}
\end{eqnarray*}

Now, using the same argument as in the proof of Proposition 4.2 from \cite%
{MaWi2}, pages 9-10, we have that%
\begin{eqnarray*}
\mathbb{E}\left[ (\mathcal{D}_{l}-\mathcal{D}_{l;m})^{2}\right]
&=&\sum_{q=m}^{\infty }\left\{ \frac{J_{2q+1}}{(2q+1)!}\right\} ^{2}\mathbb{E%
}\left[ h_{l;2q+1}^{2}\right] \\
&=&\frac{1}{l^{2}}\sum_{q=m}^{\infty }a_{q}c_{2q+1}+o(l^{-2}) \\
&\leq &\frac{1}{2\sqrt{\pi }}\frac{1}{l^{2}}\sum_{q=m}^{\infty }\frac{c_{5}}{%
q^{3/2}}+o(l^{-2})=O\left(\frac{1}{l^{2}\sqrt{m}}\right),
\end{eqnarray*}%
so that%
\begin{equation*}
\frac{2}{Var(\mathcal{D}_{l})}\mathbb{E}\left[ (\mathcal{D}_{l}-\mathcal{D}%
_{l;m})^{2}\right] =O\left(\frac{1}{\sqrt{m}}\right),
\end{equation*}%
and%
\begin{eqnarray*}
\frac{Var(\mathcal{D}_{l;m})}{Var(\mathcal{D}_{l})}+1-\frac{2\sqrt{Var(%
\mathcal{D}_{l;m})}}{\sqrt{Var(\mathcal{D}_{l})}} &=&2+O(\frac{1}{\sqrt{m}}%
)-2\sqrt{1+O\left(\frac{1}{\sqrt{m}}\right)} \\
&=&O\left(\frac{1}{\sqrt{m}}\right).
\end{eqnarray*}%
It follows immediately that
\begin{equation}
\label{2ago}
\mathbb{E}\left[ \left( \frac{\mathcal{D}_{l}}{\sqrt{Var(\mathcal{D}_{l})}}-
\frac{\mathcal{D}_{l;m}}{\sqrt{Var(\mathcal{D}_{l;m})}}\right) ^{2}\right]
=O\left(\frac{1}{\sqrt{m}}\right).
\end{equation}
Now, for every fixed $m$ we have
\begin{equation*}
\frac{\mathcal{D}_{l;m}}{\sqrt{Var(\mathcal{D}_{l;m})}}\rightarrow_{d}
\mathcal{N}(0,1)\text{ as }l\rightarrow \infty,
\end{equation*}
and since $m$ can be chosen arbitrarily large, the
random variables $\left\{ \frac{\mathcal{D}_{l}}{\sqrt{Var(\mathcal{D}_{l})}}
\right\} $ must have the same limit, bearing in mind \eqref{2ago} (see e.g. \cite{Sodin&Tsirelson}).
\end{proof}

\section{Appendix}

The Legendre polynomials are defined by Rodrigues' formula%
\begin{equation*}
P_{l}(t):=\frac{1}{2^{l}l!}\frac{d^{l}}{dt^{l}}(t^{2}-1)^{l}\text{ .}
\end{equation*}%
Legendre polynomials are orthogonal with respect to the constant weight $%
\omega (t)\equiv 1$ on $[-1,1]$, indeed%
\begin{equation}
\int_{-1}^{1}P_{l_{1}}(t)P_{l_{2}}(t)dt=\frac{2\delta _{l_{1}}^{l_{2}}}{%
2l_{1}+1}\text{ ;}  \label{legorth}
\end{equation}%
they also satisfy the well-known Hilb's asymptotics (see e.g. \cite{szego},
formula (8.21.17) on page 197):
\begin{equation}
P_{l}(\cos {\theta })=\left( \frac{\theta }{\sin {\theta }}\right)
^{1/2}J_{0}((l+1/2)\theta )+\delta (\theta )\text{ },
\label{eq:Hilb asymptotics}
\end{equation}%
where $J_{0}$ is the standard Bessel function, and the error term satisfies
\begin{equation*}
\delta (\theta )\ll
\begin{cases}
\theta ^{1/2}O(l^{-3/2})\text{ }, & cl^{-1}<\theta <\pi /2 \\
\theta ^{2}\text{ }, & 0<\theta <cl^{-1}.%
\end{cases}%
\end{equation*}%
In particular, for $\theta \in \left[ 0,\frac{\pi }{2}\right] $,
\begin{equation}
P_{l}(\cos {\theta })\ll \frac{1}{\sqrt{l\theta }}.
\label{eq:Legendre decay sqrt}
\end{equation}

Let us now review briefly some notation on Wigner's $3j$ coefficients; see
\cite{VIK}, \cite{VMK} and \cite{BieLou} for a much more detailed
discussion, in particular concerning the relationships with the quantum
theory of angular momentum and group representation properties of $SO(3).$
We start from the analytic expression (valid for $m_{1}+m_{2}+m_{3}=0,$ see
\cite{VMK}, expression (8.2.1.5))
\begin{align*}
\left(
\begin{array}{ccc}
l_{1} & l_{2} & l_{3} \\
m_{1} & m_{2} & m_{3}%
\end{array}%
\right) & :=(-1)^{l_{1}+m_{1}}\sqrt{2l_{3}+1}\left[ \frac{%
(l_{1}+l_{2}-l_{3})!(l_{1}-l_{2}+l_{3})!(l_{1}-l_{2}+l_{3})!}{%
(l_{1}+l_{2}+l_{3}+1)!}\right] ^{1/2} \\
& \times \left[ \frac{(l_{3}+m_{3})!(l_{3}-m_{3})!}{%
(l_{1}+m_{1})!(l_{1}-m_{1})!(l_{2}+m_{2})!(l_{2}-m_{2})!}\right] ^{1/2} \\
& \times \sum_{z}\frac{(-1)^{z}(l_{2}+l_{3}+m_{1}-z)!(l_{1}-m_{1}+z)!}{%
z!(l_{2}+l_{3}-l_{1}-z)!(l_{3}+m_{3}-z)!(l_{1}-l_{2}-m_{3}+z)!}\text{,}
\end{align*}%
where the summation runs over all $z$'s such that the factorials are
non-negative. This expression becomes much neater for $m_{1}=m_{2}=m_{3}=0,$
where we have%
\begin{equation*}
\left(
\begin{array}{ccc}
l_{1} & l_{2} & l_{3} \\
0 & 0 & 0%
\end{array}%
\right) =
\end{equation*}%
\begin{equation}
\left\{
\begin{array}{c}
0\text{ , for }l_{1}+l_{2}+l_{3}\text{ odd} \\
(-1)^{\frac{l_{1}+l_{2}-l_{3}}{2}}\frac{\left[ (l_{1}+l_{2}+l_{3})/2\right] !%
}{\left[ (l_{1}+l_{2}-l_{3})/2\right] !\left[ (l_{1}-l_{2}+l_{3})/2\right] !%
\left[ (-l_{1}+l_{2}+l_{3})/2\right] !}\left\{ \frac{%
(l_{1}+l_{2}-l_{3})!(l_{1}-l_{2}+l_{3})!(-l_{1}+l_{2}+l_{3})!}{%
(l_{1}+l_{2}+l_{3}+1)!}\right\} ^{1/2}\text{ } \\
\text{for }l_{1}+l_{2}+l_{3}\text{ even}%
\end{array}%
\right. .  \label{appe}
\end{equation}%
Some of the properties to follow become neater when expressed in terms of
the so-called Clebsch-Gordan coefficients, which are defined by the
identities (see \cite{VMK}, Chapter 8)%
\begin{equation}
\left(
\begin{array}{ccc}
l_{1} & l_{2} & l_{3} \\
m_{1} & m_{2} & -m_{3}%
\end{array}%
\right) =(-1)^{l_{3}+m_{3}}\frac{1}{\sqrt{2l_{3}+1}}%
C_{l_{1}-m_{1}l_{2}-m_{2}}^{l_{3}m_{3}}  \label{clewig1}
\end{equation}%
\begin{equation}
C_{l_{1}m_{1}l_{2}m_{2}}^{l_{3}m_{3}}=(-1)^{l_{1}-l_{2}+m_{3}}\sqrt{2l_{3}+1}%
\left(
\begin{array}{ccc}
l_{1} & l_{2} & l_{3} \\
m_{1} & m_{2} & -m_{3}%
\end{array}%
\right) \text{.}  \label{clewig2}
\end{equation}%
We have the following orthonormality conditions:
\begin{eqnarray}
\sum_{m_{1},m_{2}}C_{l_{1}m_{1}l_{2}m_{2}}^{lm}C_{l_{1}m_{1}l_{2}m_{2}}^{l^{%
\prime }m^{\prime }} &=&\delta _{l}^{l\prime }\delta _{m}^{m\prime },
\label{ortho1} \\
\sum_{l,m}C_{l_{1}m_{1}l_{2}m_{2}}^{lm}C_{l_{1}m_{1}^{\prime
}l_{2}m_{2}^{\prime }}^{lm} &=&\delta _{m_{1}}^{m_{1}^{\prime }}\delta
_{m_{2}}^{m_{2}^{\prime }}\text{.}  \label{ortho2}
\end{eqnarray}%
Now recall the general formula (\cite{VMK}, eqs. 5.6.2.12-13, or \cite%
{MaPeCUP}, eqs 3.64 and 6.46)
\begin{equation*}
\int_{S^{2}}Y_{l_{1}m_{1}}(x)...Y_{l_{n}m_{n}}(x)dx
\end{equation*}%
\begin{eqnarray}
&=&\sqrt{\frac{4\pi }{2l_{n}+1}}\sum_{L_{1}...L_{n-3}}\sum_{M_{1}...M_{n-3}}%
\left[
C_{l_{1}m_{1}l_{2}m_{2}}^{L_{1}M_{1}}C_{L_{1}M_{1}l_{3}m_{3}}^{L_{2}M_{2}}...C_{L_{n-3}M_{n-3}l_{n-1}m_{n-1}}^{l_{n},-m_{n}}\right.
\notag \\
&&\times \left. \sqrt{\frac{\prod_{i=1}^{n-1}(2l_{i}+1)}{(4\pi )^{n-1}}}%
\left\{
C_{l_{1}0l_{2}0}^{L_{1}0}C_{L_{1}0l_{3}0}^{L_{2}0}...C_{L_{n-3}0l_{n-1}0}^{l_{n}0}\right\} %
\right] \text{.}  \label{gaunt}
\end{eqnarray}%
Two important special cases are provided by
\begin{equation*}
\int_{\mathbb{S}^{2}}Y_{lm_{1}}(x)Y_{lm_{2}}(x)Y_{lm_{3}}(x)dx
\end{equation*}%
\begin{equation*}
=(-1)^{l-m_{3}}\sqrt{\frac{(2l+1)}{4\pi }}%
C_{l0l0}^{l0}C_{lm_{1}lm_{2}}^{l-m_{3}}
\end{equation*}%
\begin{equation}
=\sqrt{\frac{(2l+1)^{3}}{4\pi }}\left(
\begin{array}{ccc}
l & l & l \\
0 & 0 & 0%
\end{array}%
\right) \left(
\begin{array}{ccc}
l & l & l \\
m_{1} & m_{2} & m_{3}%
\end{array}%
\right) \text{ ,}  \label{gaunt3}
\end{equation}%
\ and%
\begin{equation*}
\int_{\mathbb{S}^{2}}Y_{lm_{1}}(x)Y_{lm_{2}}(x)Y_{lm_{3}}(x)Y_{lm_{4}}(x)dx
\end{equation*}%
\begin{equation*}
=\frac{(2l+1)}{\sqrt{4\pi }}\sum_{L}(-1)^{L-M}\left\{ C_{l0l0}^{L0}\right\}
^{2}\frac{C_{lm_{1}lm_{2}}^{LM}C_{lm_{3}lm_{4}}^{L,-M}}{2L+1}
\end{equation*}%
\begin{equation}
=\sqrt{\frac{(2l+1)^{4}}{4\pi }}\sum_{L}(2L+1)\left(
\begin{array}{ccc}
l & l & L \\
0 & 0 & 0%
\end{array}%
\right) ^{2}\left(
\begin{array}{ccc}
l & l & L \\
m_{1} & m_{2} & M%
\end{array}%
\right) \left(
\begin{array}{ccc}
L & l & l \\
M & m_{3} & m_{4}%
\end{array}%
\right) \text{ .}  \label{gaunt4}
\end{equation}%
Similarly, the following identities hold:%
\begin{equation*}
\int_{0}^{1}P_{l}^{3}(t)dt=\left(
\begin{array}{ccc}
l & l & l \\
0 & 0 & 0%
\end{array}%
\right) ^{2},
\end{equation*}%
\begin{equation*}
\int_{0}^{1}P_{l}^{4}(t)dt=\sum_{L=0}^{2l}(2L+1)\left(
\begin{array}{ccc}
l & l & L \\
0 & 0 & 0%
\end{array}%
\right) ^{2}.
\end{equation*}%
Finally, we recall the useful identity, valid for all $x_{1},x_{2}\in
\mathbb{S}^{2}$
\begin{equation}
P_{l}(\left\langle x_{1},x_{2}\right\rangle )=\frac{4\pi }{2l+1}%
\sum_{m=-l}^{l}Y_{lm}(x_{1})\overline{Y}_{lm}(x_{2}),
\label{exploiting}
\end{equation}%
which allows to express Legendre polynomials in terms of spherical harmonics.

\end{document}